\RequirePackage{lineno}
\documentclass[%
 superscriptaddress,
aps,pra,
reprint
]{revtex4-1}

\usepackage{stmaryrd}
\usepackage{color}
\usepackage{array}
\usepackage{enumitem}
\usepackage{mathpazo}
\usepackage{slashbox,booktabs}
\usepackage[T1]{fontenc}
\usepackage{color,graphicx}
\usepackage{times}
\usepackage{mathtools}
\mathtoolsset{centercolon}
\usepackage{setspace}
\usepackage{cleveref}
\usepackage{graphicx}
\usepackage{dcolumn}
\usepackage{bm}
\usepackage{graphicx}
\usepackage{epsfig}
\usepackage{epsf}
\usepackage{amssymb}
\usepackage{amsmath}
\usepackage{amsthm}
\usepackage{mathtools}
\usepackage{multirow}
\usepackage{cases}
\usepackage[colorlinks=true,linkcolor=blue,citecolor=blue,pdfauthor={ },pdftitle={ },pdfsubject={ },pdfkeywords={ }]{hyperref}
\usepackage{pgfplots}
\usepackage{lipsum}
\usepackage{times}

\newtheorem{definition}{Definition}
\newtheorem{theorem}{Theorem}
\newtheorem{observation}{Observation}
\newtheorem{proposition}{Proposition}

\newtheorem{lemma}{Lemma}
\newtheorem{corollary}{Corollary}

\newcommand*{\ket}[1]{\ensuremath{|#1\rangle}} %
\newcommand*{\bra}[1]{\ensuremath{\langle#1|}} %

\newcommand*{\op}[2]{\ensuremath{|#1\rangle
    \langle#2|}} %
\newcommand*{\ip}[2]{\ensuremath{\langle#1|#2\rangle}} %
\DeclareMathOperator{\tr}{tr} %
\DeclareMathOperator{\Tr}{Tr} %

\def\H{\mathcal{H}} %
\def\C{\mathbb{C}}
\def\I{\mathbb{I}}
\def\T{\mathbb{T}}
\def\L{\mathcal{L}}
\begin{document}
\title{Quantum State Tomography for Generic Pure States}

\author{Shilin Huang}
\affiliation{Institute for Interdisciplinary Information Sciences, Tsinghua University, Beijing, China}

\author{Jianxin Chen}	
\affiliation{Aliyun Quantum Laboratory, Hangzhou, China}

\author{Youning Li}
\affiliation{Department of Physics and Collaborative Innovation Center of Quantum Matter,
Tsinghua University, Beijing, China}

\author{Bei Zeng}
\affiliation{Department of Mathematics \& Statistics, University of Guelph, Guelph, Ontario, Canada}
\affiliation{Institute for Quantum Computing, University of Waterloo, Waterloo, Ontario, Canada}

\begin{abstract}
We examine the problem of whether a multipartite pure quantum state can be uniquely determined by its reduced density matrices.
We show that a generic pure state in three party Hilbert space $\H_A \otimes \H_B \otimes \H_C$, where $\mathrm{dim}(\H_A) = 2$ and $\mathrm{dim}(\H_B) = \mathrm{dim}(\H_C)$, can be uniquely determined by its reduced states on subsystems $\H_A \otimes \H_B$ and $\H_A \otimes \H_C$. Then we generalize the conclusion to the case that $\mathrm{dim}(\H_1) > 2$.
As a corollary, we show that a generic $N$-qudit pure quantum state is uniquely determined by only $2$ of its $\lceil\frac{N+1}{2}\rceil$-particle reduced density matrices. Furthermore, our results do indicate a method to uniquely determine a generic $N$-qudit pure state of dimension $D=d^N$ with only $O(D)$ local measurements, which is an improvement comparing to the previous known approach using $O(D\log^2 D)$ or $O(D\log D)$ local measurements.
\end{abstract}

\maketitle
\renewcommand\theequation{\arabic{section}.\arabic{equation}}
\setcounter{tocdepth}{4}
\makeatletter
\@addtoreset{equation}{section}
\makeatother
\section{Introduction}

The task of reconstructing the state of a quantum system by physical measurements,
known as quantum state tomography (QST) \cite{Paris2004J},
is of great importance in quantum information science, such as validating quantum devices and benchmarking \cite{lvovsky2009continuous,haffner2006scalable,baur2012benchmarking,d2002quantum,leibfried2005creation}.
However, for an $N$-qudit system, the dimension of the state space and, indeed, the number of local measurements for QST, grow exponentially
as $N$ increases.
This is a fundamental difficulty in performing QST on many-body system. In the past decade, tremendous effort has been devoted to boost the efficiency of QST \cite{vanner2012cooling,flammia2012quantum,gross2010quantum,lu2016tomography,cramer2010efficient,baldwin2016strictly,kosaka2009spin,wang2016quantum} under different situations.


Naive full tomography, as literally, need $D^{2}$ measurements, where $D=d^N$ is the dimension of the $N$-qudit system ($d$ is the dimension of the qudit).
For pure states or density operators with low rank, reduced density matrices (RDMs) may contain information to characterize the global system\cite{Chen2017Physical,Chen2017Joint}, thus a more promising approach for QST can be implemented as follows: acquiring a set of RDMs by local measurements, then applying post-processing algorithm to reconstruct the state \cite{linden2002almost,linden2002parts,diosi2004three,jones2005parts,chen2012comment,chen2012ground,chen2013uniqueness}. Usually, we have no knowledge about the system before measurements. Therefore, to adopt this approach, one important criterion must be satisfied, that is, given results of chosen local measurements,
the original state should be Uniquely Determined among All states (UDA).
It was first proved in \cite{linden2002almost} that a generic three-qubit pure state is UDA by its two-body RDMs. \cite{jones2005parts} provided a more general result, showing that a generic $N$-qudit pure quantum state is UDA by its $\left( \lceil \frac{N}{2} \rceil + 1\right)$-body RDMs. \cite{jones2005parts} also proved that $\lfloor \frac{N}{2} \rfloor$-body RDMs are not sufficient for uniquely determining a generic $N$-qudit pure state.
However, a gap appears when $N$ is odd: whether $\frac{N+1}{2}$-qudit RDMs are sufficient to uniquely determine a generic pure state or not remains open.
\cite{chen2013uniqueness} proved that a generic $3$-qudit pure state can be uniquely determined by its $2$-qudit RDMs. Hence it is natural to conjecture that for every odd $N \ge 3$, $\frac{N+1}{2}$-qudit RDMs are sufficient to uniquely determine a generic pure state. In this paper, the conjecture is proved correct and thus the gap is filled.

Furthermore, we hope the set of local measurements to be as small as possible. \cite{gross2010quantum} gave a protocol to reconstruct low rank $r$ density matrices of $N$-qubit system using $O(rD\log^2D)$ Pauli measurements ($D = 2^N$). In particular, when the original state is pure, i.e., $r = 1$, $O(D\log^2D)$ Pauli measurements are needed.
\cite{jones2005parts} provided an approach using $\lfloor \frac{N}{2} \rfloor$ $\left( \lceil \frac{N}{2} \rceil + 1\right)$-qudit RDMs and hence
$O(D \log D)$ ($D = d^N$) local measurements to uniquely determine a generic $N$-qudit pure state. However, a simple parameter counting argument suggests that $O(D)$ measurements could be possibly sufficient for uniquely determining an $N$-qudit pure state. In this paper, we provide a valid scheme to uniquely determine a generic $N$-qudit pure state, using only $O(D)$ local measurements.

This paper is organized as follows. In the beginning, we consider a three body system $\H_A \otimes \H_B \otimes \H_C$, where $\textrm{dim}\left(\H_B\right) = \textrm{dim}\left(\H_C\right)$. We show in section \ref{sec2} that when $\textrm{dim}\left( \H_A \right) = 2$, a generic pure state is equivalent to a `triangular state' under local unitary transformations. In section \ref{sec4}, we prove that almost all triangular states can be uniquely determined by its
RDMs on subsystems $\H_A \otimes \H_B$ and $\H_A \otimes \H_C$. The remaining proof is in section \ref{sec5}, which complete the proof by claiming that, the set of states that can not be transformed into a triangular form, or those not UDA via its RDMs on subsystems $\H_A \otimes \H_B$ and $\H_A \otimes \H_C$, has measure zero.

\section{Triangular Form for states in $\C^2 \otimes \C^d \otimes \C^d$} \label{sec2}

Before showing that a generic state is UDA by its RDMs, we have a crucial observation that unique determinability is invariant under local unitary (LU) transformations. Thus we can focus on those states in canonical form under LU transformations. To be more precise, we introduce several notation and concepts, which will be used in later discussion:

\begin{definition}
Let $\ket{\psi}, \ket{\psi'} \in \C^{d_A} \otimes \C^{d_B} \otimes \C^{d_C}$ be two state vectors.
Say $\ket{\psi}$ and $\ket{\psi'}$ are LU-equivalent, if and only if there exist unitary operators $U_1 \in U(d_A)$, $U_2 \in U(d_B)$, $U_3 \in U(d_C)$
such that
\begin{equation}
\ket{\psi'} = \left(U_1 \otimes U_2 \otimes U_3\right) \ket{\psi}.
\end{equation}
\end{definition}

\begin{definition}
Let $\rho, \rho' \in \L\left( \C^{d_A} \otimes \C^{d_B} \otimes \C^{d_C} \right)$ be two self-adjoint operators.
Say $\rho$ and $\rho'$ are equivalent (denoted by $\rho \simeq \rho'$), if and only if their reduced density matrices (RDM) satisfy
\begin{equation}
\Tr_{B} (\rho) = \Tr_B(\rho'), \quad \Tr_{C}(\rho) = \Tr_{C}(\rho').
\end{equation}
\end{definition}

\begin{definition}
Let $\ket{\psi} \in \C^{d_A} \otimes \C^{d_B} \otimes \C^{d_C}$ be a vector.
Say $\ket{\psi}$ is UDA, if and only if for every positive semidefinite operator $\rho \in \L\left( \C^{d_A} \otimes \C^{d_B} \otimes \C^{d_C} \right)$, $\rho \simeq \op{\psi}{\psi}$ implies $\rho = \op{\psi}{\psi}$.
\end{definition}

To see why our observation holds, we have the following lemma:
\begin{lemma}\label{lma1}
Let $\ket{\psi} \in \C^{d_A} \otimes \C^{d_B} \otimes \C^{d_C}$ be a UDA state, $U_A \in U(d_A)$, $U_B \in U(d_B)$ and $U_C \in U(d_C)$ be three unitary transformations. Then the state $\ket{\psi'} = \left(U_A \otimes U_B \otimes U_C\right) \ket{\psi}$ is also UDA.
\end{lemma}
\begin{proof}
Suppose there exists another $\rho' \ne \op{\psi'}{\psi'}$ such that $\rho' \simeq \op{\psi'}{\psi'}$. Let $\rho = \Large(U_A^{\dagger} \otimes U_B^{\dagger} \otimes U_C^{\dagger}\Large)\rho'\Large(U_A \otimes U_B \otimes U_C\Large)$. We have
$\Tr_{C} (\rho) = \left(U_A^{\dagger} \otimes U_B^{\dagger}\right) \Tr_{C}(\rho') \left(U_A \otimes U_B\right)
=\Large(U_A^{\dagger} \otimes U_B^{\dagger}\Large) \Tr_{C}(\op{\psi'}{\psi'}) \Large(U_A \otimes U_B\Large) = \Tr_{C}(\op{\psi}{\psi}).$
Similarly, we have $Tr_{B}(\rho) = \Tr_{B}\left(\op{\psi}{\psi}\right)$, which contradicts to the assumption that $\ket{\psi}$ is UDA.
\end{proof}

It is well-known that for bipartite quantum system, the canonical form is the so-called Schmidt decomposition \cite{schmidt1989theorie}. For multipartite systems, generalized versions of Schmidt decomposition and classifications of multipartite entanglement have been deeply studied from various aspects \cite{acin2000generalized,carteret2000multipartite}. In this paper, we mainly focus on the case of $\C^2\otimes \C^d\otimes \C^d$ system. We will adopt the lemma used in \cite{chen2017} to study the generalized Schmidt decomposition of $\C^2\otimes \C^d\otimes \C^d$ quantum system. For sake of the readability, we include the proof for self-contained reading.

\begin{lemma}[Lemma~5, \cite{chen2017}]\label{CanonicalForm}
A generic 3-body state $\ket{\psi} \in \C^2 \otimes \C^d \otimes \C^d$, is LU-equivalent to a state
\begin{equation}
\ket{\phi} = \sum_{i=1}^2\sum_{j=1}^d\sum_{k=1}^d \phi_{ijk} \ket{i^A}\ket{j^B} \ket{k^C}
\end{equation}
such that
\begin{enumerate}
\item $\phi_{ijk} = 0$ when $j > k$,
\item $\phi_{1ij}$ is real when $j \le i+1$.
\end{enumerate}
We call $\ket{\phi}$ the triangular form of $\ket{\psi}$.
\end{lemma}

Consider the following linear isomorphism $\varphi: \C^2 \otimes \C^d \otimes \C^d \rightarrow \L(\C^d) \times \L(\C^d)$, such that
\begin{eqnarray}
&& \varphi\Big(\ket{1^A} \ket{i^B}  \ket{j^C}\Big) = \Big( \op{i}{j}, 0 \Big),\notag\\
&& \varphi\Big( \ket{2^A}  \ket{i^B}  \ket{j^C} \Big) = \Big( 0, \op{i}{j} \Big).
\end{eqnarray}
For a state $\ket{\psi} \in \C^2 \otimes \C^d \otimes \C^d$ and two unitary operator $U, V \in U(d)$ acting on the second and third subsystems respectively.
Assuming $\varphi(\ket{\psi}) = (A, B)$, we have
\begin{equation}
\varphi\Big( \left(\I_A \otimes U \otimes V\right) \ket{\psi} \Big) = \Big( UAV^{\dagger}, UBV^{\dagger} \Big).
\end{equation}
If $UAV^{\dagger}$ and $UBV^{\dagger}$ are both upper triangular matrices, $\I_A \otimes U \otimes V \ket{\psi}$ would be in the triangular form of $\ket{\psi}$. The following proposition gives a method on how to find such $U$ and $V$, and it implies Lemma \ref{CanonicalForm} directly.

\begin{proposition}\label{triangular}
Let $A, B \in \L(\C^d)$ be $2$ linear operators, if the polynomial $\mathrm{det}(A-xB)$ has degree $d$, then there exists $U,V \in U(d)$ such that (i) both $UAV^{\dagger}$ and $UBV^{\dagger}$ are upper triangular matrices. (ii) $\bra{i}UAV^{\dagger}\ket{j}$ is real when $j \le i+1$.

\end{proposition}
\begin{proof}
When $d = 1$, the statement is trivial. Suppose the statement holds when $d = n-1$.
Let $A,B \in \L(\C^n)$ be two linear operators such that the polynomial $f(x) =\mathrm{det}\left(A - xB\right)$ has degree $n$.
Let $x_1 \in \C$ be one of the root of $f(x)=0$. Since $\mathrm{det}\left(A-x_1B\right) = 0$, $\mathrm{ker}\left(A-x_1B\right)$ is not empty.

Let $\ket{v} \in \mathrm{ker}\left(A-x_1B\right)$ be a normalized vector and denote $\ket{w} = B\ket{v}$. We immediately have $\ket{w} \ne 0$, otherwise, for every $x \in \C$, we have
$\left(A - xB\right) \ket{v} = (x_1-x)B \ket{v} = (x_1 - x) \ket{w} = 0$,
which implies that $f(x) \equiv 0$, contradiction.

Let $V_1, W_1 \in U(n)$ be arbitrary unitary operators such that
$V_1\ket{v} = \ket{1}$ and $W_1 \ket{w} \in \mathrm{span}\left( \left\{\ket{1} \right\} \right)$
Then we have
\begin{eqnarray*}
W_1 A V_1^{\dagger} \ket{1} = \alpha x_1 \ket{1},\quad  W_1 B V_1^{\dagger} \ket{1} = \alpha \ket{1}
\end{eqnarray*}
for some $\alpha \in \C$.
Define $P = \sum_{i=2}^n \op{i}{i}$, then we can define $2$ operators $A', B' \in \L(\C^{n-1})$, such that
\begin{eqnarray*}
A' = P W_1 A V_1^{\dagger} P,\quad B' = P W_1 B V_1^{\dagger} P.
\end{eqnarray*}
It is easily to check that
\begin{equation}
\mathrm{det}\left[W_1(A-xB)V_1^{\dagger}\right] = \alpha(x_1-x)\mathrm{det}\left(A' - xB'\right).
\end{equation}
Thus, the degree of the polynomial $\mathrm{det}\left(A'-xB'\right)$ must be $n-1$.

By our assumption, there must exist $U_2, V_2 \in U(n-1)$ such that both $U_2A'V_2^{\dagger}$ and $U_2B'V_2^{\dagger}$ are upper triangular matrices.
Suppose
\begin{eqnarray*}
\bra{1} U_1 A V_1^{\dagger} \ket{2} = r_1 e^{i\theta_1},\quad \bra{1} U_1 A V_1^{\dagger} \ket{1} = r_2 e^{i\theta_2},
\end{eqnarray*}
where $r_1, \theta_1, r_2, \theta_2 \in \mathbb{R}$.

Let
\begin{eqnarray}
U &=& \left(e^{-i\theta_1}\op{1}{1} + U_2\right) U_1, \notag\\
V &=& \left(e^{i(\theta_1-\theta_2)}\op{1}{1} + V_2\right) V_1.
\end{eqnarray}
It is easy to check that $U,V$ satisfies the statement.
\end{proof}

For a generic $3$-body pure state $\ket{\psi} \in \C^2 \otimes \C^d \otimes \C^d$, the polynomial $\mathrm{det}(A-xB)$ is of degree $d$, where $(A, B) = \varphi(\ket{\psi})$.  By Proposition \ref{triangular}, simultaneous upper triangulation of $A$ and $B$ can always be realized and thus Lemma \ref{CanonicalForm} is proved.

\section{Regular Triangular Vector} \label{sec4}

For every state vector $\ket{\psi} \in \C^2 \otimes \C^d \otimes \C^d$, its triangular form $\ket{\phi}$ lies in the manifold
\begin{eqnarray*}
&\ &\T_d \equiv \bigg\{ \sum_{i=1}^2 \sum_{j=1}^d \sum_{k=1}^d \phi_{ijk} \ket{i^A} \ket{j^B} \ket{k^C} \bigg|\\
&\ &\phi_{ijk} = 0\ \text{when}\ j>k, \phi_{1jk} \in \mathbb{R}\ \text{when}\ k\le j+1\bigg\}.
\end{eqnarray*}
Every vector $\ket{\psi} \in \T_d$ can be expanded in computational basis as
\begin{equation}
\ket{\psi} = \sum_{i=1}^d \sum_{j=i}^d \ket{\psi_{ij}^A}\ket{i^B}\ket{j^C},
\end{equation}
where $\ket{\psi_{ij}} \in \C^2$ and do not have to be orthogonal to each other. Now we consider a subset of $\T_d$, whose complement is measure zero set:
\begin{definition}
We call $\ket{\psi}$ a \textbf{regular triangular vector} if $\ket{\psi}$ satisfies the following condition
\begin{enumerate}
\item $\forall\, 1 \le i \le j \le d$, $\ket{\psi_{ij}^A} \ne 0$,
\item $\forall\, 1 \le i \le j \le d$, $1 \le k \le d$, $\ket{\psi_{ij}^A} \in \mathrm{span}\left( \left\{\ket{\psi_{kk}^A} \right\}\right)$ implies  $i = j = k$.
\end{enumerate}
\end{definition}

\begin{lemma}\label{genregtri}
A generic state vector in $\T_d$ is regular triangular.
\end{lemma}

\begin{proof}
For every $1 \le i\le j \le d$ and $k \in [d]$, define
$$\eta_{ij} \equiv \left\{ \ket{\psi} \in \T_d: \ket{\psi_{ij}^A} = 0 \right\},$$
and
$$\xi_{ijk} \equiv \left\{\begin{array}{lr}
\emptyset & i = j = k\\
\left\{\ket{\psi} \in \T_d: \ket{\psi_{ij}^A} \in \mathrm{span}\left(\{\ket{\psi_{kk}^A} \}\right) \right\}  & \textrm{otherwise}
\end{array} \right.
$$
Obviously, $\eta_{ij}$ and $\xi_{ijk}$ are all measure-zero sets, and hence the set
\begin{eqnarray*}
\Delta &\equiv& \left\{ \ket{\psi}: \ket{\psi}\textrm{ is not regular triangular} \right\}\\
&=& \left( \bigcup_{1 \le i\le j \le d}
\eta_{ij} \right) \bigcup \left( \bigcup_{1\le i\le j, k} \xi_{ijk}\right)
\end{eqnarray*}
also has measure zero.
\end{proof}

\subsection{Regular Triangular Vector is UDA}
Let $$\ket{\psi} = \sum_{i=1}^d \sum_{j=1}^d \ket{\psi_{ij}^A} \ket{i^B}\ket{j^C}$$ be a vector in $\T_d$. Consider the projective measurement $\left\{P_1, P_2\right\}$ on the third subsystem, where
\begin{eqnarray*}
P_1 = \op{d^C}{d^C},\ P_2 = \I_C - \op{d^C}{d^C}.
\end{eqnarray*}
The vector
$$P_2 \ket{\psi} = \sum_{i=1}^{d-1} \sum_{j=i}^{d-1} \ket{\psi_{ij}^A}\ket{i^B} \ket{j^C}$$
lies in the manifold $\T_{d-1}$.
If $\ket{\psi}$ is a regular triangular state in $\T_d$, $P_2\ket{\psi}$ is also a regular triangular vector.
Moreover, the following theorem holds:
\begin{theorem}\label{reduce}
Let $\ket{\psi} \in \T_d$ be a regular triangular vector. For every positive semidefinite operator $\rho \simeq \op{\psi}{\psi}$, we have
\begin{enumerate}
\item $P_1 \rho P_1 = P_1\op{\psi}{\psi}P_1$,
\item $P_2 \rho P_2 \simeq P_2\op{\psi}{\psi}P_2$.
\end{enumerate}
\end{theorem}

Note that we have the following corollary of Theorem \ref{reduce}:
\begin{corollary}\label{lemma2}
All regular triangular vectors in $\T_d$ are UDA.
\end{corollary}
\begin{proof}

The statement is trivial when $d = 1$.
Suppose that all regular triangular vectors in $\T_{d-1}$ are UDA.
Consider an arbitrary regular triangular vector $\ket{\psi}$ in $\T_d$.
For every positive semidefinite operator $\rho \simeq \op{\psi}{\psi}$, by Proposition \ref{reduce}, we have
\begin{equation}
P_1\rho P_1 = P_1\op{\psi}{\psi} P_1, \quad P_2 \rho P_2 \simeq P_2\op{\psi}{\psi} P_2.
\end{equation}
By our assumption that all regular triangular vectors in $\T_{d-1}$ are UDA, we have $P_2 \rho P_2 = P_2 \op{\psi}{\psi} P_2$
since $P_2\ket{\psi}$ is a regular triangular vector.
Therefore we can view $\rho$ as an unnormalized mixed state in the subspace spanned by $\left\{ P_1\ket{\psi}, P_2\ket{\psi} \right\}$, i.e.,
we can expand $\rho$ as
\begin{eqnarray}
\rho &=&P_1\op{\psi}{\psi}P_1 + P_2\op{\psi}{\psi}P_2\nonumber\\
     &+& \alpha P_1\op{\psi}{\psi}P_2 + \alpha^\ast P_2 \op{\psi}{\psi} P_1. \label{expansion}
\end{eqnarray}

Tracing out the second particle and doing further calculation on (\ref{expansion}), we have
\begin{eqnarray}
\bra{1^C} \Tr_B(\rho) \ket{d^C}
= \alpha^\ast\  \Tr_B\left(\ip{1^C}{\psi}\ip{\psi}{d^C}\right). \label{eqn35}
\end{eqnarray}
On the other hand, by original assumption $\Tr_B(\rho)=\Tr_B(\op{\psi}{\psi})$, we have
\begin{eqnarray}
\Tr_B(\bra{1^C}\rho\ket{1^C}) = \Tr_B\left(\ip{1^C}{\psi}\ip{\psi}{d^C}\right) \ne 0. \label{eqn36}
\end{eqnarray}
Comparing (\ref{eqn35}) and (\ref{eqn36}),
we must have $\alpha = 1$. Therefore $\rho = \op{\psi}{\psi}$, i.e.,  $\ket{\psi}$ is UDA.
\end{proof}

\subsection{Proof of Theorem \ref{reduce}}
Before proving Theorem \ref{reduce}, we introduce a useful observation at first:

\begin{observation}\label{Propo1}
Let $\H_A$ and $\H_B$ be two finite dimensional Hilbert spaces,
$H \in \L(\H_A \otimes \H_B)$ be a positive semi-definite operator, and $\ket{n} \in \H_A$ be a normalized vector.
\begin{itemize}
\item[Case A.] If $\bra{n} H \ket{n} = \op{r}{r}$ for some non-zero vector $\ket{r} \in \H_B$, then there exists a unique state $\ket{\varphi} \in \H_A \otimes \H_B$ such that $H \ket{n} = \op{\varphi}{r}$ and $H - \op{\varphi}{\varphi} \ge 0$.
\item[Case B.] If $\bra{n} H \ket{n} = 0$, then $H \ket{n} = 0$.
\end{itemize}
\end{observation}

\begin{proof}

If $\bra{n} H \ket{n} = \op{r}{r}$ for some non-zero vector $\ket{r} \in \H_B$,
then for every $\ket{s} \in \H_B$ such that $\ip{s}{r} = 0$, we have $\bra{ns} H \ket{ns} = 0$ and thus $H \ket{ns} = 0$.
We have
\begin{eqnarray*}
H \ket{n} &=& H \ket{n} \left[ \frac{1}{\ip{r}{r}} \op{r}{r} + \left(\I_B - \frac{1}{\ip{r}{r}}\op{r}{r} \right) \right]\\
&=& \frac{1}{\ip{r}{r}}H \op{nr}{r}.
\end{eqnarray*}
Define $\ket{\varphi} = \frac{1}{\ip{r}{r}}H \ket{nr}$, then $H \ket{n} = \op{\varphi}{r}$.

For every vector $\ket{\psi} \in \H_A \otimes \H_B$,
we have
\begin{eqnarray}
&\ &\bra{\psi} \left(H - \op{\varphi}{\varphi} \right) \ket{\psi}\nonumber\\
&=& \bra{\psi} H \ket{\psi} - \frac{\left| \bra{\psi} H \ket{nr} \right|^2}{\ip{r}{r}^2}\nonumber\\
&=& \frac{\bra{\psi} H \ket{\psi} \bra{nr} H \ket{nr} - \left| \bra{\psi} H \ket{nr} \right|^2}{\ip{r}{r}^2}\ge 0.
\end{eqnarray}
Note that the last step is given by Cauchy-Schwarz inequality. Therefore $H - \op{\varphi}{\varphi} \ge 0$.

If $\bra{n}H\ket{n} = 0$, then for every normalized state $\ket{b} \in \H_B$ and $\ket{m} \in \H_A \otimes \H_B$ such that $\ip{m}{nb} = 0$, we must have $\bra{m}H\ket{nb} = 0$. Otherwise,
\begin{equation*}
\mathrm{det}\bigg[\left(\op{nb}{nb} + \op{m}{m}\right) H \left(\op{nb}{nb} + \op{m}{m}\right)\bigg]
\end{equation*}
will be negative, which contradicts to the assumption that $H \ge 0$. Thus,
\begin{equation*}
H\ket{nb} = \left(\I_{AB} - \op{nb}{nb}\right) H \ket{nb} + \op{nb}{nb}H\ket{nb} = 0
\end{equation*}
for all $\ket{b}$, i.e., $H\ket{n} = 0$.
\end{proof}

As the first step to prove Theorem \ref{reduce}, we prove the following proposition.
\begin{proposition} \label{prop1}
$\bra{d^B} \rho \ket{d^B} = \ket{\psi_{dd}^A}\bra{\psi_{dd}^A} \otimes \ket{d^C}\bra{d^C}$.
\end{proposition}
\begin{proof}
Define $\ket{\zeta_i} \equiv \sum_{j=i}^d \ket{\psi_{ij}^A}\ket{j^C}$ and $H = \tr_2(\rho) - \bra{d^B}\rho\ket{d^B}.$
We first prove that
$H \ge \sum_{i=1}^{d-1} \op{\zeta_i}{\zeta_i}$.

Suppose $H \ge \sum_{i=1}^{k-1} \op{\zeta_i}{\zeta_i}$ for some $1 \le k < d$.
Let $H' = H - \sum_{i=1}^{k-1} \op{\zeta_i}{\zeta_i}$. We have $$H' + \bra{d^B}\rho\ket{d^B} = \sum_{i=k}^d \op{\zeta_i}{\zeta_i}.$$
Note that
\begin{equation}
\bra{k^C} \Big( H' + \bra{d^B}\rho\ket{d^B} \Big) \ket{k^C} = \ket{\psi_{kk}^A}\bra{\psi_{kk}^A},
\end{equation}
since $H' \ge 0$, $\bra{d^B} \rho \ket{d^B} \ge 0$, and $\op{\psi_{kk}^A}{\psi_{kk}^A}$ has rank $1$,
there must exist $\alpha \ge 0$ such that
$$\bra{d^Bk^C} \rho \ket{d^Bk^C} = \alpha\ket{\psi_{kk}^A}\bra{\psi_{kk}^A}.$$
However, notice that
\begin{equation}
\sum_{i=1}^d \bra{d^Bi^C} \rho \ket{d^Bi^C} = \ket{\psi_{dd}^A}\bra{\psi_{dd}^A},
\end{equation}
there must exist $\beta \ge 0$ such that $$\bra{d^Bk^C}\rho\ket{d^Bk^C} = \beta \ket{\psi_{dd}^A}\bra{\psi_{dd}^A}.$$
Since $\ket{\psi}$ is regular triangular, for every $k < d$, we have $\ket{\psi_{kk}^A} \notin \mathrm{span}\left( \left\{ \ket{\psi_{dd}^A} \right\} \right)$.
Thus $\alpha = \beta = 0$, i.e., $$\bra{d^Bk^C} \rho \ket{d^Bk^C} = 0.$$
By Observation \ref{Propo1}, $\bra{d^B} \rho \ket{d^Bk^C} = 0$, hence
\begin{equation*}
H'\ket{k^C} = \Big( \sum_{i=k}^d \op{\zeta_i}{\zeta_i} - \bra{d^B}\rho\ket{d^B} \Big) \ket{k^C} = \ket{\zeta_k}\bra{\psi_{kk}^A}.
\end{equation*}
Moreover,
\begin{equation*}
H' - \op{\zeta_k}{\zeta_k} = H - \sum_{i=1}^{k} \op{\zeta_k}{\zeta_k} \ge 0.
\end{equation*}
Since $H \ge \sum_{i=1}^{k-1} \op{\zeta_i}{\zeta_i}$ holds trivially when $k = 1$,
we can prove $H \ge \sum_{i=1}^{d-1} \op{\zeta_i}{\zeta_i}$ by induction.

Up to now, we have
\begin{eqnarray*}
&\ &\bra{d^B} \rho \ket{d^B} = \sum_{i=1}^{d} \op{\zeta_i}{\zeta_i} - H\\
&\le& \op{\zeta_d}{\zeta_d} = \ket{\psi_{dd}^A}\bra{\psi_{dd}^A} \otimes \ket{d^C}\bra{d^C}.
\end{eqnarray*}
Since
\begin{eqnarray*}
&\  & \Tr_C\left(\bra{d^B} \rho \ket{d^B}\right) = \ket{\psi_{dd}^A}\bra{\psi_{dd}^A}\\
& = &\Tr_C\left(\ket{\psi_{dd}^A}\bra{\psi_{dd}^A} \otimes \ket{d^C}\bra{d^C}\right),
\end{eqnarray*}
we must have
\begin{equation}
\bra{d^B} \rho \ket{d^B} = \ket{\psi_{dd}^A}\bra{\psi_{dd}^A} \otimes \ket{d^C}\bra{d^C}.
\end{equation}
\end{proof}

By Proposition \ref{prop1}, we have
$$\bra{d^Bi^C}\rho \ket{d^Bi^C} = \delta_{id}\ket{\psi_{dd}^A}\bra{\psi_{dd}^A}.$$
By Observation \ref{Propo1}, for every $i<d$, we have $\bra{i^C} \rho \ket{d^Bi^C} = 0$.
Therefore
\begin{eqnarray*}
&\ &\bra{d^C} \rho \ket{d^Bd^C}\\
 &=& \left(\Tr_C(\rho) - \sum_{i=1}^{d-1} \bra{i^C} \rho \ket{i^C} \right) \ket{d^B}\\
  &=& \Tr_C(\rho) \ket{d^B}\\
&=&\left( \sum_{i} \ket{\psi_{id}^A}\ket{i^B} \right) \bra{\psi_{dd}^A}\\
 &=& \ip{d^C}{\psi} \bra{\psi_{dd}^A}.
\end{eqnarray*}
Moreover,
\begin{eqnarray*}
\bra{d^C} \rho \ket{d^C} &\ge& \ip{d^C}{\psi}\ip{\psi}{d^C}.
\end{eqnarray*}

Since $\Tr_B\left(\bra{d^C} \rho \ket{d^C}\right) = \Tr_B\left(\ip{d^C}{\psi}\ip{\psi}{d^C}\right)$,
we have $\bra{d^C} \rho \ket{d^C} = \ip{d^C}{\psi}\ip{\psi}{d^C}$.
Recall that $P_1 = \ket{d^C}\bra{d^C}$, we have
\begin{eqnarray*}
P_1 \rho P_1 = \ket{d^C} \ip{d^C}{\psi}\ip{\psi}{d^C}\bra{d^C}
= P_1 \op{\psi}{\psi} P_1.
\end{eqnarray*}
 Up to now, we have finished the proof of the first part of Theorem \ref{reduce}.

Consider the operator $P_2 \rho P_2$. For every $i < d$, we have $\bra{d^Bi^C}\rho \ket{d^Bi^C}= 0$.
By Observation \ref{Propo1}, we have $\rho \ket{d^Bi^C} = 0$. Furthermore, for every $1 \le i,j \le d$, we have
\begin{equation*}
\bra{d^Bi^C} \rho \ket{d^Bj^C} = 0,
\end{equation*}
i.e., $\rho \in \L\left(\H_p \otimes \H_{d-1} \otimes \H_{d-1}\right)$.

Since $P_2$ acts trivially on the second subsystem, we have
\begin{eqnarray}\label{equation21}
\Tr_B\left(P_2 \rho P_2\right) = \Tr_B \left(P_2 \op{\psi}{\psi} P_2\right).
\end{eqnarray}
On the other hand,
\begin{eqnarray}\label{theorem22}
&\ &\Tr_C(P_2\rho P_2)\nonumber\\
 &=& \Tr_C(\rho) - \bra{d^C} \rho \ket{d^C} \nonumber\\
&=& \Tr_C(\op{\psi}{\psi}) - \ip{d^C}{\psi}\ip{\psi}{d^C} \nonumber\\
 &=& \Tr_C\left( P_2 \op{\psi}{\psi} P_2 \right).
\end{eqnarray}
 Combining (\ref{equation21}) and (\ref{theorem22}), we conclude that $$P_2 \rho P_2 \simeq P_2 \op{\psi}{\psi} P_2$$ and thus the second part of Theorem \ref{reduce} has been proved.

\section{A Generic Pure State in $\C^p \otimes \C^d \otimes \C^d$ is $(12,13)$-UDA} \label{sec5}

\subsection{The case of $p = 2$}

\begin{lemma}
Almost all pure states in $\C^2 \otimes \C^d \otimes \C^d$ are UDA.
\end{lemma}

\begin{proof}
Consider the smooth mapping $$M: \T_d \times SU(d) \times SU(d) \times U(1) \rightarrow \C^2 \otimes \C^d \otimes \C^d$$ defined as
\begin{equation}
M\left( \ket{\phi}, U, V, e^{i\theta} \right) = e^{i\theta} \left( \I_2 \otimes U \otimes V \right) \ket{\phi}.
\end{equation}

Define $$\Delta_1 \equiv \C^2 \otimes \C^d \otimes \C^d \Large\setminus M\Big[ \T_d \times SU(d) \times SU(d) \times U(1) \Big]$$
as the set of states that do not have a triangular form, and
\begin{align*}
\Delta_2 \equiv \Delta \times SU(d) \times SU(d) \times U(1).
\end{align*}
Here we recall that $$\Delta = \left\{ \ket{\psi} \in \T_d: \ket{\psi}\textrm{ is not regular triangular} \right\}.$$
By Lemma \ref{lma1}, the set $$\Delta_3 \equiv \left\{ \ket{\psi} \in \C^2 \otimes \C^d \otimes \C^d: \ket{\psi}\ \textrm{is not UDA} \right\}$$
is the subset of $\Delta_1 \cup M(\Delta_2)$.

By Lemma \ref{CanonicalForm}, $\Delta_1$ has measure-zero. Since $\Delta$ has measure-zero in $\T_d$ by Lemma \ref{genregtri}, we evidently conclude that $\Delta_2$ has measure-zero in $\T_d \times SU(d) \times SU(d) \times U(1)$.
The dimension of the manifold $\T_d$ is
\begin{equation}
 \frac{d(d+1)}{2} \times 2 \times 2 - d - (d-1) = 2d^2+1,
\end{equation}
and indeed,
\begin{eqnarray}
&\ &\mathrm{dim}\Big( \T_d \times SU(d) \times SU(d) \times U(1) \Big)\nonumber\\
&=& 2d^2+1 + 2(d^2-1) + 1\nonumber\\
&=& 4d^2 = \mathrm{dim} \left(\C^2 \otimes \C^d \otimes \C^d\right).
\end{eqnarray}
By Theorem 6.9 in \cite{Lee2003}, the set $M(\Delta_2)$ has measure-zero in $\C^2 \otimes \C^d \otimes \C^d$.
Therefore the set $\Delta_3 \subseteq \Delta_1 \cup M(\Delta_2)$ has measure-zero in $\C^2 \otimes \C^d \otimes \C^d$.

\end{proof}





\subsection{General Case}
\begin{theorem}\label{final}
For $p \ge 2$, almost all pure states in $\C^p \otimes \C^d \otimes \C^d$ are UDA.
\end{theorem}
\begin{proof}
We have proved the case of $p = 2$.
When $p > 2$, 
let
\begin{equation}
P_{ij} = \ket{i^A}\bra{i^A} + \ket{j^A}\bra{j^A}, \quad i \ne j.
\end{equation}
For a positive semidefinite operator $\rho \in \L\left( \C^{p} \otimes \C^d \otimes \C^d \right)$, we have
$$\Tr_B\left(P_{ij}\rho P_{ij}\right) =P_{ij}\Tr_B\left(\rho\right)P_{ij}$$
and
$$\Tr_C\left(P_{ij}\rho P_{ij}\right) =P_{ij}\Tr_C\left(\rho\right)P_{ij}$$ for all $i \ne j$.

Since $P_{ij}$ is a full-rank linear mapping from $\C^p \otimes \C^d \otimes \C^d$ to
$\C^2 \otimes \C^d \otimes \C^d$, the set
\begin{equation*}
\left\{ \ket{\psi} \in \C^p \otimes \C^d \otimes \C^d :\ P_{ij} \ket{\psi}\ \textrm{is\ not\ UDA\ for\ some\ }i,j\right\}
\end{equation*}
has measure-zero.

Let $\ket{\psi} \in \C^p \otimes \C^d \otimes \C^d$ be a generic state vector.
For every positive semidefinite operator $\rho$
such that $\rho \simeq \op{\psi}{\psi}$, we have
\begin{equation}
\Tr_r \left( P_{ij} \rho P_{ij} \right) = P_{ij} \Tr_r\left(\op{\psi}{\psi}\right) P_{ij},\ r = B,C.
\end{equation}
Note that $P_{ij} \ket{\psi}$ is UDA. Thus
$P_{ij}\rho P_{ij} = P_{ij}\op{\psi}{\psi} P_{ij}$.

We have
\begin{eqnarray*}
&\ &\bra{i^A} \rho \ket{j^A} = \bra{i^A} P_{ij}\rho P_{ij} \ket{j^A}\\
&=& \bra{i^A} P_{ij}\op{\psi}{\psi} P_{ij} \ket{j^A} =\ip{i^A}{\psi}\ip{\psi}{j^A},
\end{eqnarray*}
and
\begin{eqnarray*}
&\ &\bra{i^A} \rho \ket{i^A} = \bra{i^A} P_{ij}\rho P_{ij} \ket{i^A}\\
&=& \bra{i^A} P_{ij}\op{\psi}{\psi} P_{ij} \ket{i^A}_1 =\ip{i^A}{\psi}\ip{\psi}{i^A}
\end{eqnarray*}
for all $i \ne j$. Thus, $\rho = \op{\psi}{\psi}$.
\end{proof}

\begin{corollary}\label{corollary}
Let $\H \equiv \H_1 \otimes \ldots \otimes \H_N$ be an $N$-qudit system, where $\H_i \cong \C^d$. A generic pure state in $\H$
can be uniquely determined by only $2$ of its $\lceil \frac{N+1}{2} \rceil$-body RDMs.
\end{corollary}
\begin{proof}
Take $m = \lfloor \frac{N-1}{2} \rfloor$, and split the whole Hilbert space into $3$ parts,
\begin{eqnarray*}
\H_A &=& \H_1 \otimes \ldots \otimes \H_{N-2m}\\
\H_B &=& \H_{N-2m+1} \otimes \ldots \otimes \H_{N-m},\\
\H_C &=& \H_{N-m+1} \otimes \ldots \otimes \H_{N}.
\end{eqnarray*}

Then $\H$ can be considered as a three-body system $\H_A \otimes \H_B \otimes \H_C$, with $\mathrm{dim}\left(\H_B\right) = \mathrm{dim}\left(\H_C\right)$. By theorem \ref{final}, a generic pure state $\ket{\psi}$ in $\H$ can be uniquely determined by
$\Tr_B\left(\op{\psi}{\psi}\right)$ and
$\Tr_C\left(\op{\psi}{\psi}\right)$, that are $2$ of its $\lceil \frac{N+1}{2} \rceil$-body RDMs.
\end{proof}

We have proved that a generic $N$-qudit state is UDA by $2$ of its $\lceil \frac{N+1}{2} \rceil$-body RDMs. Indeed, if we need to reconstruct a generic $N$-qudit pure state via local measurements, $2 \times d^{2\lceil \frac{N+1}{2} \rceil} = O(d^N)$ local measurements are sufficient.


\section*{Acknowledgments}
We thank Jie Zhou for helpful discussions.
Y.~L. acknowledges support from Chinese Ministry of Education under grants No.20173080024. B.~Z. is supported by he Natural Sciences and Engineering Research Council of Canada (NSERC) and Canadian Institute for Advanced Research(CIFAR).
\bibliography{uda}

\end{document}